\documentclass[a4paper,UKenglish]{article}

\usepackage{microtype}
\usepackage{fullpage}
\usepackage{graphics}
\graphicspath{{./figures/},{./figures/}}

\usepackage{authblk}

\usepackage{graphicx}
\usepackage{amsthm}
\usepackage{amsmath}
\usepackage{enumerate}
\usepackage{color,subcaption}
\usepackage{xspace}
\usepackage{url}

\usepackage{amsfonts}\usepackage{amssymb}
\usepackage{thmtools}\usepackage{mathscinet}
\usepackage{thm-restate}
    \newtheorem{theorem}{Theorem}
    
    \newtheorem{lemma}[theorem]{Lemma}

    	\definecolor{darkgreen}{rgb}{0.01, 0.93, 0.29}
\definecolor{lightbrown}{rgb}{0.91, 0.4, 0.11}
\usepackage{framed}
\usepackage[disable]{todonotes}

\title{StreamTable: An Area Proportional Visualization for  Tables with Flowing Streams\footnote{Work is supported in part by  the Natural Sciences and Engineering Research Council of Canada (NSERC), and by two CFREF grants coordinated by GIFS and GIWS.}}

\author{Jared Espenant}
\author{Debajyoti Mondal}
\affil{Department of Computer Science\\
University of Saskatchewan, Saskatoon, Saskatchewan, Canada \\
  \texttt{jae608@usask.ca, d.mondal@usask.ca}}
 
\begin{document}

\maketitle

\begin{abstract}
Let $M$ be a two-dimensional  table with  each cell weighted by a nonzero positive number. A StreamTable visualization of $M$ represents the columns as non-overlapping vertical streams and the rows as horizontal stripes such that the intersection between a stream and a stripe is a rectangle with area equal to the weight of the corresponding cell.  To avoid large wiggle of the streams, it is desirable to keep the consecutive cells in a stream to be adjacent. Let $B$ be the smallest axis-aligned bounding box containing the StreamTable. Then the difference between the area of $B$ and the sum of the weights is referred to as the excess area. We attempt to optimize various StreamTable aesthetics (e.g., minimizing excess area, or maximizing  cell adjacencies in streams).
\begin{itemize}
    \item If the row  permutation is fixed and the row heights are given, then we give an $O(rc)$-time algorithm  
    to optimizes these aesthetics, where $r$ and $c$ are the number of rows and columns, respectively.  
    \item If the row permutation is fixed but the row heights can be chosen, then we discuss a technique to compute an aesthetic (but not necessarily optimal) StreamTable by solving a  quadratically-constrained quadratic program, 
    followed by iterative improvements. 
    If the row heights are restricted to be integers, then we prove the problem to be NP-hard. 
    \item If the row  permutations can be chosen, then we show that it is NP-hard to find a row permutation that optimizes  the area or adjacency aesthetics.
\end{itemize}
\end{abstract}

\section{Introduction}
Proportional area charts and cartographic visualizations commonly represent data values as geometric objects. Table cartogram~\cite{DBLP:journals/comgeo/EvansFKKMNV18} is a brilliant way to  
visualize tables as cartograms, where each table cell is mapped to a convex quadrilateral with area equal to the cell's weight. Furthermore, the visualization preserves cell adjacencies and the quadrilaterals are packed together in a rectangle with no empty space in between (e.g., see Figure~\ref{fig:intro}(e)). However, since the cells in a table cartogram are represented with convex quadrilaterals, it may sometimes become difficult to follow the rows and columns~\cite{rakib}. This motivated us to look for a solution, where each row is  represented with a horizontal \emph{stripe} (i.e., a region bounded by two horizontal lines) and the cells in each row are  represented with axis aligned rectangles inside the corresponding stripe.

Streamgraphs are examples where the columns can be thought of as vertical stripes. Given a set of variables, a \emph{streamgraph} visualizes how their values change over time by representing each variable with a   flowing river-like stream (e.g., an  $x$-monotone polygon). The width of the stream at a timestamp is determined by the value of the variable at that time. Figure~\ref{fig:intro}(a) illustrates a streamgraph with five variables. Streamgraphs are often used to create infographics of temporal data~\cite{DBLP:journals/tvcg/ByronW08}, e.g., box office revenues for movies~\cite{DBLP:journals/cgf/BartolomeoH16}, various statistics or demographics of a population over time~\cite{DBLP:journals/tvcg/HavreHWN02}, etc. 

In this paper, we introduce StreamTable that extends this idea of a streamgraph to visualize   tables or spreadsheets. We now  formally define a StreamTable.

\subsection{StreamTable}
Let $T$ be an $r \times c$ table with $r$ rows and $c$ columns, where each cell is weighted by a nonzero positive number. A \emph{StreamTable} visualization of $T$ is a partition of an axis-aligned rectangle $R$ into $r$ consecutive horizontal stripes that represent the rows of $T$, where each stripe is further divided into  rectangles to represent the cells of its corresponding row. A column $q$ of $T$ is thus  represented by a sequence of rectangles  corresponding to the cells of $q$. By a \emph{stream} we refer to such a sequence of rectangles that represents a column of $T$. Furthermore, a StreamTable must satisfy the  following properties.
\begin{figure}[pt]
  \centering
  \includegraphics[width=\textwidth]{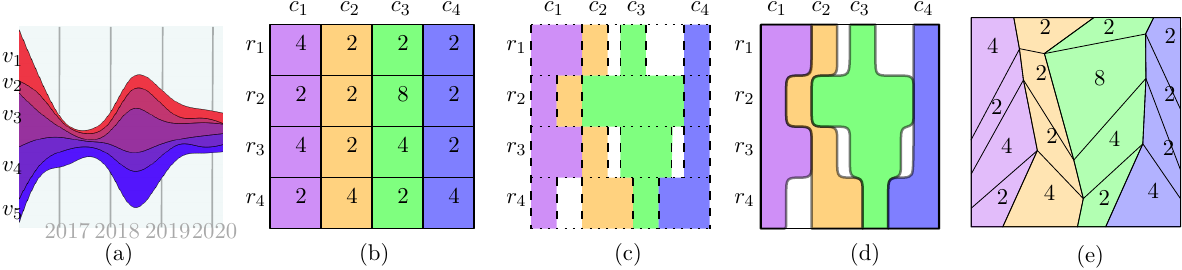}
  \caption{(a) A streamgraph. (b) A table $T$. (c) A StreamTable for $T$. (d) A StreamTable visualization with smooth streams. (e) A table cartogram for $T$.}
  \label{fig:intro}
\end{figure}
\begin{enumerate}
    \item[$P_1.$] The left side of the leftmost stream (resp., the right side of the rightmost stream) must be aligned to the left side (resp., right side) of $R$.
    \item[$P_2.$] For each cell of $T$, the area of its corresponding rectangle in the StreamTable must be equal to the cell's weight. 
\end{enumerate}
Property $P_1$ ensures an aesthetic alignment with the row labels and provides a sense of total visualization area. Property $P_2$ provides an area proportional representation of the table cells. 
Figure~\ref{fig:intro}(b) illustrates a table and Figure~\ref{fig:intro}(c) illustrates a corresponding StreamTable. The stripes (rows) are shown in dotted lines and the partition of the stripes are shown in dashed lines. Figure~\ref{fig:intro}(d) illustrates an aesthetic visualization of the streams 
after smoothing the corners.

Note that a StreamTable may contain rectangular regions that do not correspond to any cell. We refer to such regions as \emph{empty regions} and the sum of the area of all empty regions as the \emph{excess area}. While computing a StreamTable, a natural optimization criterion is to minimize this excess area. However, minimizing excess area may sometimes result into disconnected streams. Figure~\ref{fig:intro2}(b) illustrates a StreamTable with zero excess area, where the consecutive rectangles for column $c_2$ are not adjacent (i.e., no two consecutive rectangles of $c_2$ share a common boundary point).  If a pair of cells are consecutive in a column but the corresponding rectangles are nonadjacent in the stream, then they \emph{split} the stream. To maintain the stream connectedness, it is desirable to minimize the number of such splits. As illustrated in Figure~\ref{fig:intro2}(c)-(d), one may choose non-uniform row heights or reorder the rows to optimize the aesthetics. Such reordering operations also appear in matrix reordering problems~\cite{DBLP:conf/diagrams/MakinenS00} where the goal is to reveal clusters in matrix data. StreamTable computation also relates to  floorplanning~\cite{DBLP:conf/ispd/ChenF98,DBLP:journals/mmor/Rosenberg89} and   area-universal rectangular layout problems~\cite{DBLP:journals/siamcomp/EppsteinMSV12,DBLP:journals/jocg/BuchinELNS16}, where the horizontal adjacencies are not mandatory but vertical adjacencies must be preserved. 

\subsection{Our Contribution} We explore StreamTable from a theoretical perspective and consider the following two problems. 

\begin{figure}[pt]
  \centering
  \includegraphics[width=.9\textwidth]{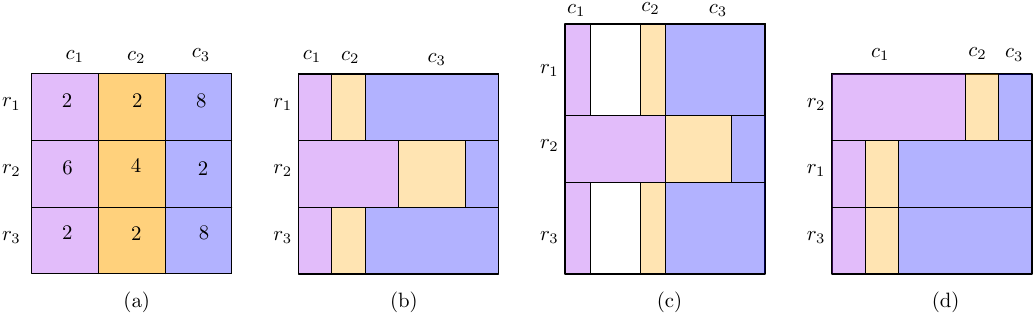}
  \caption{(a) A table.  (b) A StreamTable with no excess area and 2 splits. (c) A StreamTable with non-uniform row heights, non-zero excess area, but no split. (d) A StreamTable with no excess area and 1 split (obtained by reordering rows).}
  \label{fig:intro2}
\end{figure}


 \smallskip
\noindent\fbox{%
    \parbox{.98\columnwidth}{%
\textbf{Problem 1 (StreamTable with no Split, Minimum Excess Area, and Fixed Row Ordering).} Given an $r\times c$ table $T$, can we compute a StreamTable for $T$ in polynomial time with no split and minimum excess area? Note that in this problem, the StreamTable must respect the row ordering of $T$. }}

\smallskip
 If the row heights are restricted to be integers, then we show the problem to be NP-hard. In general, the problem can be modeled leveraging a quadratically-constrained quadratic program, and a solution  computed by non-linear programming solver may be iteratively improved by adjusting the row heights. However, this only provides a heuristic solution. While Problem 1 remains open, if the input additionally specifies a set $\{h_1,\ldots,h_r\}$ of  nonzero positive numbers  to be chosen as row heights, then we can compute a StreamTable with minimum excess area in $O(rc)$ time.    
 Since choosing a fixed row height helps to obtain a fast algorithm and to  compare the cell areas more accurately, we examined whether one can leverage the row ordering to further improve the StreamTable aesthetics.  

\smallskip
\noindent\fbox{%
    \parbox{.98\columnwidth}{%
\textbf{Problem 2 (Row-Permutable StreamTable with Uniform Row Heights).} Given a table $T$ and a non-zero positive number $\delta>0$, can we compute a StreamTable in polynomial time by setting $\delta$ as the row height, and minimizing the excess area (or, the number of splits)? Note that in this problem, the row ordering can be chosen.}}

\smallskip
We show that Problem 2 is NP-hard, i.e., we show that computing a StreamTable with no excess area and minimum number of split is NP-hard and similarly, computing a StreamTable with no split and minimum excess area is NP-hard.

\section{StreamTable (No Split, Min. Excess Area,  Fixed Row Order)}
In this section we compute StreamTables by respecting the row ordering of the input table. We first explore the case  when the row heights are given, and then the case when the row heights can be chosen.

\subsection{Fixed Row Heights}
\label{sec:lp}
Let $T$ be an $r\times c$ table and let $\{h_1,\ldots,h_r\}$ be a set of nonzero positive numbers to be chosen as row heights. We now introduce some notation for the rectangles and streams in the StreamTable. Let  $w_{i,j}$ be the weight for the $(i,j)$th entry of $T$, where $1\le i\le r$ and $1\le j\le c$, and let $R_{i,j}$ be the rectangle with height $h_i$ and width $(w_{i,j}/h_{i,j})$. Let $a_{i,j}$ and $b_{i,j}$ be the $x$-coordinates of the left and right side $R_{i,j}$.

We now show that a StreamTable $\mathcal{R}$ for $T$ with no split and minimum excess area can be constructed using a greedy algorithm $\mathcal{G}$, as follows:

\begin{enumerate}
    \item[]\textbf{Step 1.} Draw the rectangles $R_{i,1}$ of the first column such that they are left aligned. 
    \item[]\textbf{Step 2.} For each $j<c$, draw the $j$th stream by minimizing the sum of  $x$-coordinates $a_{i,j}$, but ensuring that the stream remains connected. 
    \item[]\textbf{Step 3.} Draw  the rectangles $R_{i,c}$ of the last column by minimizing the maximum $x$-coordinate over $b_{i,c}$, but ensuring that the rectangles are right aligned.  
\end{enumerate} 
 
For every column  $j$, let $A(\mathcal{R},j)$  be the orthogonal polygonal chain determined by the left  side of $R_{i,j}$. Similarly, we define (resp., $B(\mathcal{R},j)$) for the right side of $R_{i,j}$. We now have the following lemma.

\begin{lemma}
\label{lem:t}
 $\mathcal{G}$ computes a StreamTable $\mathcal{R}$ with no split and minimum excess area. 
\end{lemma}
\begin{proof}
We employ an induction on the number of columns. For an $r\times c$ table $T$ with $c=2$, it is straightforward to verify the lemma. We now assume that the lemma holds for every table with $j$ columns where $1\le j<c$. Consider now a table with $c$ columns and let $\mathcal{R}^*$ be an optimal  StreamTable with no split and minimum excess area. 

We first show that the first two streams  of $\mathcal{R}^*$ can be replaced with the corresponding streams of $\mathcal{R}$. To observe this first note that the stream for the first column must be drawn left-aligned, and since the rectangle heights are given, the right side of the streams $B(\mathcal{R},1)$ must coincide with $B(\mathcal{R}^*,1)$. Consider now the left sides of the  second streams. If $A(\mathcal{R},2)$ does not coincide with  $A(\mathcal{R}^*,2)$, then there must be non-zero area between them. Let $A$ be an orthogonal polygonal chain constructed by taking the left envelope of these two chains. In other words, for each row, we choose the part of the chain that have the  minimum $x$-coordinate. Since the streams for $\mathcal{R}$ and $\mathcal{R}^*$ are connected, the stream determined by $A$ must be connected. Since the sum of $x$-coordinates is smaller for $A$, the polygonal chain $A(\mathcal{R},2)$ must coincide with $A$. Thus the right side of the stream, i.e., the polygonal chain  $B(\mathcal{R},2)$, must remain to the left of $B(\mathcal{R}^*,2)$. 

We can now construct an $r\times (c-1)$ table $T'$ by treating the polygonal chain  $B(\mathcal{R},2)$ as $B(\mathcal{R},1)$. By induction, $\mathcal{G}$ provides a StreamTable $\mathcal{R}'$ with no split and minimum excess area. We can thus obtain the StreamTable $\mathcal{R}$ by replacing the first stream with the two streams that were constructed using the greedy approach. 
\end{proof}
We now have the following theorem. 
\begin{theorem}
Given an $r\times c$ table $T$ and a height for each row, a StreamTable $\mathcal{R}$ for $T$ with no split and minimum excess area can be computed in $O(rc)$ time such that $\mathcal{R}$ respects the row ordering of $T$.  
\end{theorem}
\begin{proof}
By Lemma~\ref{lem:t}, it suffices to show that $\mathcal{G}$ takes $O(rc)$ time. It is straightforward to compute the drawing of the first column (i.e., \textbf{Step 1}) in $O(r)$ time. To compute \textbf{Step 2}, for each column $j$, we first compute two sequence of rectangles $R_t$ and $R_b$, and then process them to find the stream for the $j$th column. 

We construct $R_t$ from top-to-bottom (e.g., see Figure~\ref{fig:z2}), and refer to a placement of the rectangle $R(i,j)$ in $R_t$ as $R_t(i,j)$. We first place $R_t(1,j)$ starting at $b(1,j-1)$, and then greedily place $R_t(i,j)$, where $i>1$, such that each rectangle lies as much to the left as possible maintaining the adjacency with the previously placed rectangle. The adjacency between subsequent rectangles  is  broken when the $b(i,j-1)$ is larger than the $x$-coordinate of the right side of the previously placed rectangle $R_t(i-1,j)$. We construct $R_b$ from bottom-to-top symmetrically. It is straightforward to compute $R_t$ and $R_b$ in $O(r)$ time. We now construct the stream  of the $j$th column by taking for each $i$, the rectangle with the larger starting $x$-coordinate among  $R_t(i,j)$ or $R_b(i,j)$. This also takes $O(r)$ time.

Let $R(i,j)$ be a rectangle in $\mathcal{R}$. We call \emph{$R(i-1,j)$ a  parent of  $R(i,j)$} if $a(i-1,j) = b(i,j)$. Similarly, we call \emph{$R(i+1,j)$ a  parent of  $R(i,j)$} if $a(i+1,j) = b(i,j)$.  By the greedy construction, every rectangle $R(i,j)$ must have a parent unless it starts at $b(i,j-1)$. If $R(i,j)$ does not have a parent, then we call it a \emph{root rectangle}.

\begin{figure}[h]
    \centering
    \includegraphics[width=.8\textwidth]{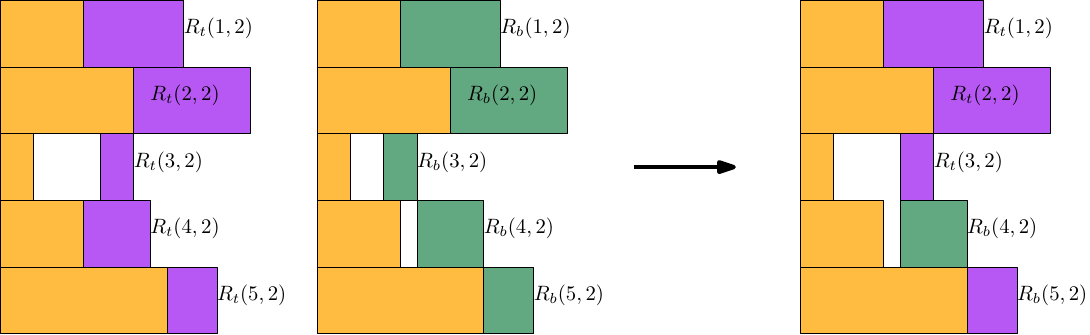}
    \caption{The computation of $R_t$ (purple) and $R_b$ (green), and the resulting second stream.}
    \label{fig:z2}
\end{figure}

Let $A(\mathcal{R},j)$ be the left side of the $j$th stream. We now show that $A(\mathcal{R},j)$ determines a connected stream and minimizes the sum of $x$-coordinates. 

\textbf{Connectedness:} Suppose for a contradiction that $R(i-1,j)$ and $R(i,j)$ are not connected, and without loss of generality assume that $R(i-1,j)$ comes from $R_t$, i.e., $R(i-1,j)= R_t(i-1,j)$. 

If $R(i,j) = R_t(i,j)$, then $b(i,j-1)$ is larger than the $x$-coordinate of the right side of $R_t(i-1,j)$. Therefore, the starting $x$-coordinate of $R_b(i,j)$ must be larger than $b(i,j-1)$ and hence  the starting $x$-coordinate of $R_b(i-1,j)$ must be larger than that of $R_t(i-1,j)$, which contradicts the assumption that  $R(i-1,j)$ comes from $R_t$.  

If $R(i,j) = R_b(i,j)$, then the starting $x$-coordinate of $R_b(i,j)$ is larger than that of $R_t(i,j)$. Furthermore, since $R(i,j)$ and $R(i-1,j)$ are not adjacent,  the starting $x$-coordinate of $R_b(i,j)$ must be larger than the right side of $R_t(i-1,j)$. Consequently, the starting $x$-coordinate of $R_b(i-1,j)$ must be larger than that of $R_t(i-1,j)$, which contradicts the assumption that  $R(i-1,j)$ comes from $R_t$.  

\textbf{Minimization:} Suppose for a contradiction that there exists a no-split  drawing $\mathcal{R'}$ for the $j$th stream such that the sum of $x$-coordinates of $A(\mathcal{R'},j)$ is smaller than that of $A(\mathcal{R},j)$. Then there must exist  a rectangle $R(i,j)$ in $R$ such that the starting $x$-coordinate of $R(i,j)$ is larger than the corresponding rectangle $R'(i,j)$ in $\mathcal{R'}$. In the following we show that such a scenario cannot exist.

If $R(i,j)$ is a root rectangle, then the starting $x$-coordinate of $R'(i,j)$ must coincide with $R(i,j)$. Hence we may assume that $R(i,j)$ has a parent. We follow the parent repeatedly until we reach a root  $R(q,j)$. Since $R(q,j)$ is a root rectangle, its starting $x$-coordinate is determined by $b(q,j-1)$. Furthermore,  the starting $x$-coordinate of $R'(q,j)$ cannot be smaller than $b(q,j-1)$. The difference   between the starting $x$-coordinates of  $R(i,j)$ and $R(q_j)$ is exactly the sum of the width of the rectangles $R(i,j),\ldots, R(q_j)$. Hence, starting $R'(i,j)$ earlier than $R(i,j)$ will split the $j$th stream in $\mathcal{R'}$, which contradicts the assumption that $\mathcal{R'}$ is a no-split drawing.

\begin{figure}[pt]
    \centering
    \includegraphics[width=.8\textwidth]{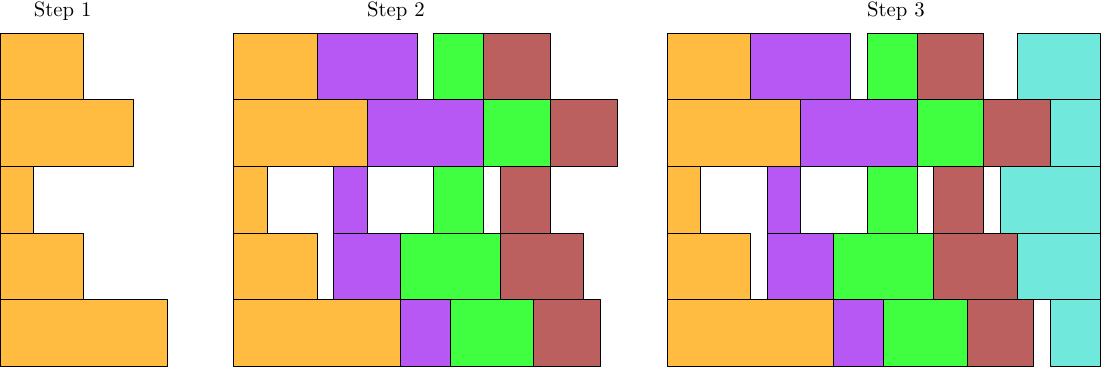}
    \caption{A simple run of the greedy algorithm, computing a no split minimum excess area StreamTable with 5 rows and columns.}
    \label{fig:z3}
\end{figure}
Finally, it is straightforward compute \textbf{Step 3}   in $O(r)$ time by first following \textbf{Step 2} and then moving the rectangles rightward to make them right aligned (e.g., see Figure~\ref{fig:z3}).  
\end{proof} 

We now show how to formulate a system of linear equations to compute a StreamTable for $T$ with no split and minimum excess area such that the height of the $i$th row is set to $h_i$,  where $1\le i\le r$.  This will be useful for the subsequent section. 
Let $d_{i,j}$ be a variable to model the adjacency between $R_{i,j}$ and $R_{i+1,j}$, where $1\le i\le r-1$ and $1\le j\le c$. We minimize the excess area: $\sum\limits_{j=1}^{r}\sum\limits_{k=1}^{c-1} h_{j}(a_{j,k+1}-b_{j,k})$, subject to the following constraints.
\begin{enumerate}
    \item $a_{j,1}=a_{j+1,1}$ and $b_{j,c}=b_{j+1,c}$, where $j=1, \ldots,r-1$. This ensures StreamTable property $P_1$.
    \item $b_{j,k}-a_{j,k} = (w_{j,k}/h_j)$, where $j=1, \ldots,r$ and $k=1, \ldots,c$. This ensures 
    property $P_2$.
    \item $a_{j,k}\le d_{j,k}\le b_{j,k}$ and $a_{j+1,k}\le d_{j,k}\le b_{j+1,k}$, where $1\le j\le r-1$ and $1\le k\le c$.  This ensures that there is no split in the streams.
\end{enumerate} 
Since $h_1,\ldots,h_r$ are fixed constants, the above system with the constraint that the variables must be non-negative can be modeled as a linear program, e.g., see  Figure~\ref{fig:GurobiExample} (left). 
 
\subsection{Variable Row Heights}
\label{vrh}

We model this case by treating  $h_1,\ldots,h_j$ as variables. Hence the objective and constraint functions yield a  quadratically-constrained quadratic program. 
Note that scaling down the height of a StreamTable by some $\delta\in (0,1]$ and scaling up the width by 
\todo{No mention of quadratic constraints. DONE.}
$1/\delta$ do not change the excess area. Therefore, a non-linear program solver may end up generating a final StreamTable with bad aspect ratio. Hence we suggest to add another constraint: $h_1+\ldots+h_k = H$, where $H$ is the desired height of the visualization.  Figure~\ref{fig:GurobiExample} (right) shows an example (not necessarily optimal) solution    computed using a non-linear program solver Gurobi~\cite{GurobiNon-Linear}. 

\begin{figure}[pt]
  \centering
  \begin{subfigure}{.48\textwidth}
    \centering
    \includegraphics[width=1.08\textwidth]{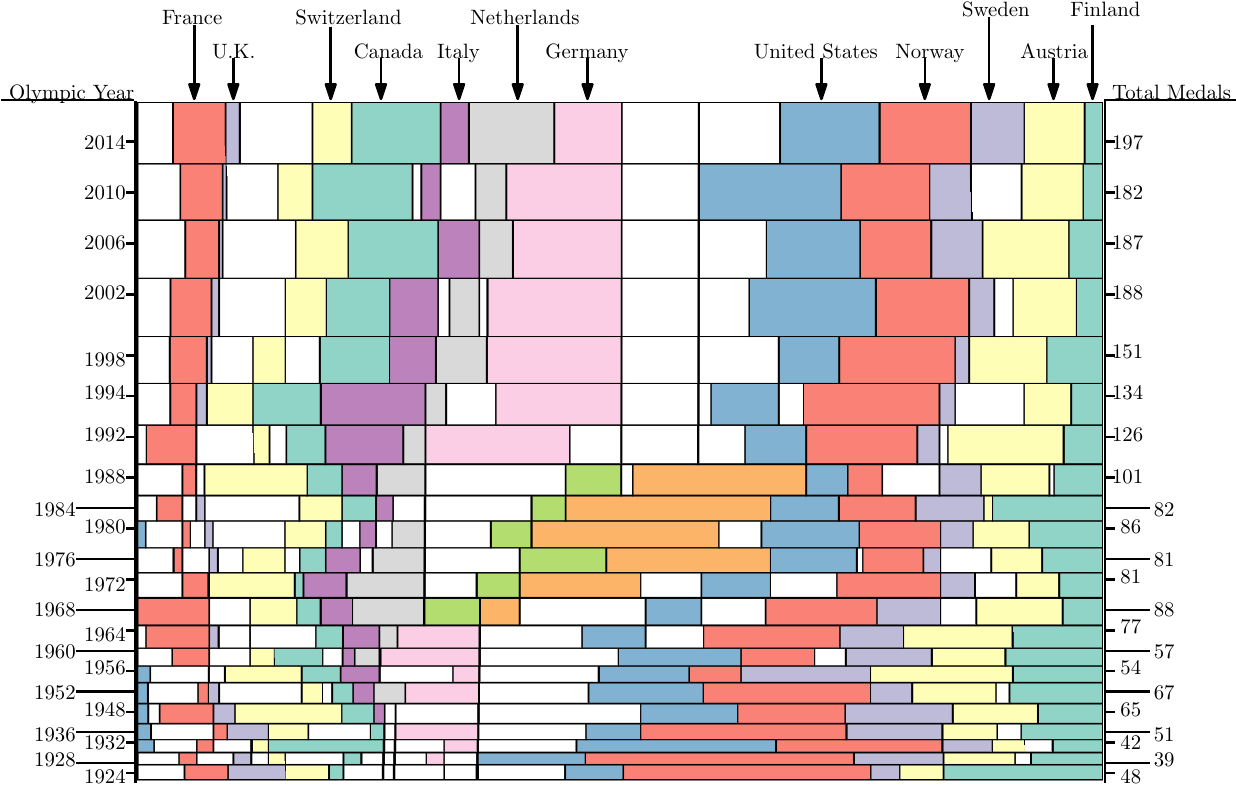}
  \end{subfigure}\hfill
  \begin{subfigure}{.48\textwidth}
    \centering
    \includegraphics[width=\textwidth]{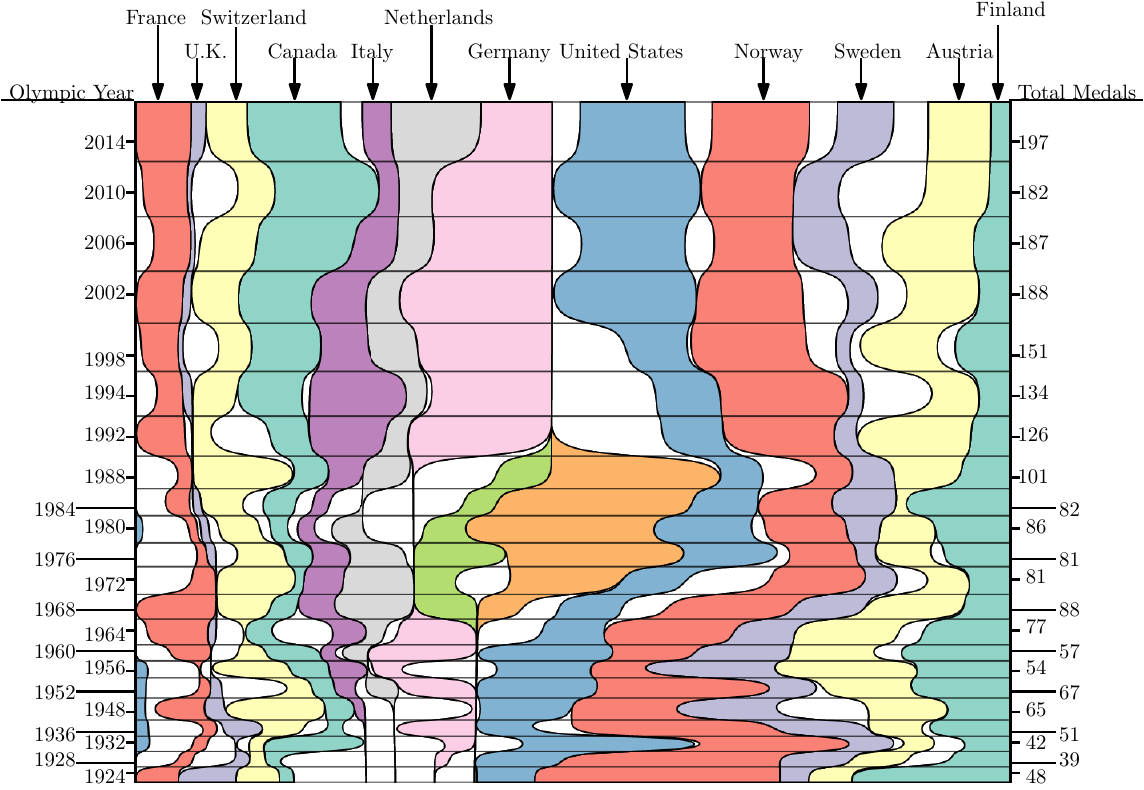}
  \end{subfigure}
\caption{StreamTables of a   Winter Olympics dataset (left) using a linear program with row height  proportional to the row sum, and (right) using Gurobi with a fixed total height and with corner smoothing. 
}
\label{fig:GurobiExample}
\end{figure}


\todo{JARED: we need to stop Gurobi a bit early, the visual difference is barely noticable.}

 
\textbf{Local Improvement:} We now show how a non-optimal StreamTable may be improved further by examining each empty cell individually, while   deciding whether that cell can be removed by shrinking the height of the corresponding row. By $E_{i,j}$ we denote the empty rectangle between the rectangles $R_{i,j}$ and $R_{i,j+1}$. We first refer the reader to Figure~\ref{fig:heuristic}(a)--(b). Assume that we want to decide whether the empty cell $E_{i,j}(=E_{2,4})$ can be removed by scaling down the height of the second row. The idea is to grow  the rectangles to the left (resp., right) of $E_{i,j}$ towards the right (resp., left) respecting the adjacencies and area. 
\begin{figure}[h]
  \centering
  \includegraphics[width=\textwidth]{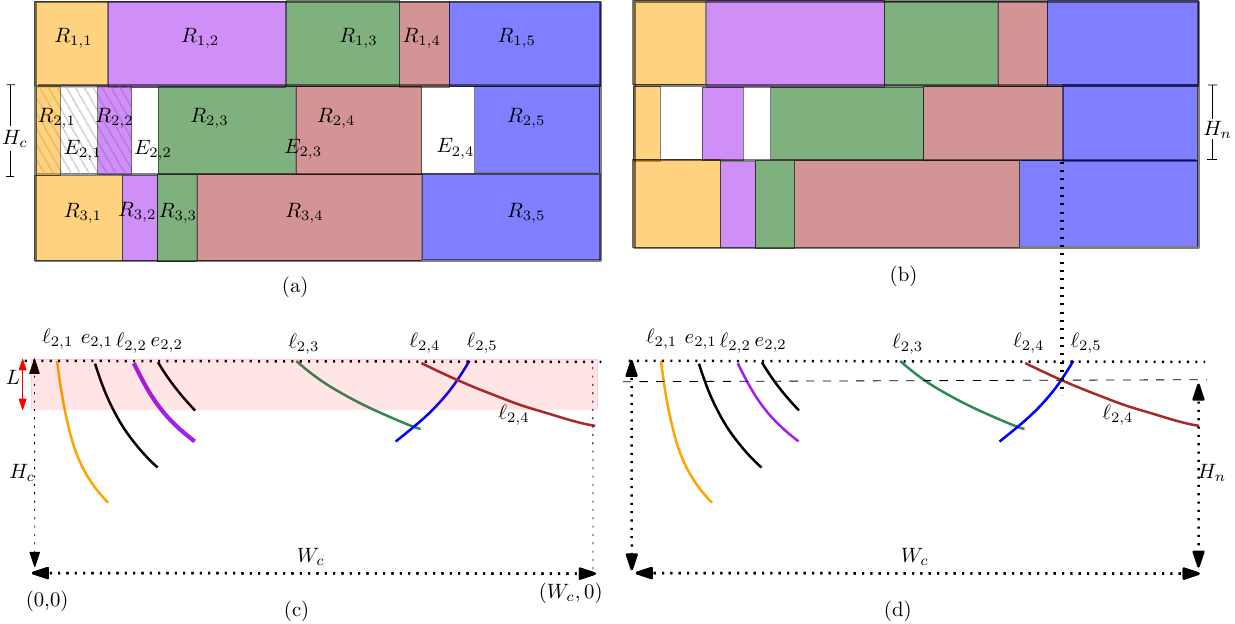}
  \caption{ (a) A StreamTable with width $W_c$ and height $H_c$. (b) Removal of the empty rectangle $E_{2,4}$ (c)--(d) Illustration for computing the new height $H_n$ of the second row.  }
  \label{fig:heuristic}
\end{figure}

Now consider a rectangle $R_{i,k}(=R_{2,2})$ before $E_{i,k}(=E_{2,4})$. Let $G_{i,k}$ be the rectangle determined by the $i$th row with left and right sides coinciding with the left and right sides of $R_{i,1}$ and $R_{i,k}$, respectively. Figure~\ref{fig:heuristic}(a) shows $G_{2,2}$ in a falling pattern. Let $\ell_{i,k}$ be the length of $G_{i,k}$. Let $A_{i,k}$ is the initial area of $G_{i,k}$, and our goal is to keep  this area fixed as we scale down the height of the $i$th row. The height   of $G_{i,k}$ is defined by  $f(\ell_{i,k}) = A_{i,k}/\ell_{i,k}$. Since the rectangles of the $(i-1)$th and $(i+1)$th rows do not move, $f(\ell_{i,k})$ does not split the $(k+1)$th stream as long as $\ell_{i,k}$ is upper  bounded by the right sides of $R_{i-1,k+1}$ and $R_{i+1,k+1}$. Figure~\ref{fig:heuristic}(c) plots these functions, where $H_c$ is the current height of the second row. The height function for $G_{2,2}$ is drawn in thick purple in the interval $[\ell_{2,2},\min\{q_{1,3},q_{3,3}\}]$, where $q_{1,3}$ and $q_{3,3}$ are the right sides of $R_{1,3}$ and $R_{3,3}$, respectively.

We construct such functions also for all the empty rectangles $E_{i,k}$, where $1\le k<j$. These are labelled with $e_{i,k}$. Finally, we construct these functions symmetrically for the rectangles that appears after  $E_{i,j}$. We then find a height $H_n$ by determining the common interval $L$ where all these functions are valid individually (Figure~\ref{fig:heuristic}(c)), and then  determining the first intersection (if any) in this interval, as illustrated in Figure~\ref{fig:heuristic}(d). If no such  intersection point exists, then we can shrink the row by an amount equal to the length of the interval $L$.  

We iterate over the empty rectangles as long as we can find an empty rectangle to improve the solution, or to a maximum number of iterations. However, this only provides a heuristic algorithm, and thus  Problem 1 remains open.

\textbf{Restricting Row Heights to Integers:} 
We now show that if the row heights are restricted to be positive integers, then we   prove the problem to be NP-hard.

\begin{theorem}
\label{thm:integer}
Given a table $T$ and a positive integer $H$, it is NP-hard  to compute a minimum-area no-split StreamTable of height $H$ with row heights as integers respecting the row ordering of $T$.
\end{theorem}
\begin{proof}
We reduce the NP-hard problem \emph{clique}~\cite{Garey1979}, where the input is a graph $G$ and a positive integer $k$ and the goal is to find a set of $k$ vertices that are pairwise adjacent. The problem remains NP-hard even when $1<k<n$.  Given an instance $G$ of the clique problem with $n$ vertices and $m$ edges, we construct a table $T$ with $n$ rows and $m$ columns as follows.

\begin{enumerate}
    \item For each edge $e \in E_{G}$, we create a column  called an \emph{edge column}, and label it by $e$.  
    \item We insert an additional column at the left and right sides of the table and also between every pair of adjacent columns.  We refer to these columns as \textit{line columns}. Each cell of a line column has a weight of $\epsilon = \frac{1}{ n(m+1)}$.  
    \item For each vertex $v \in V_{G}$, we create a row and assign it the label $v$. For each edge column, we assign each cell a weight $6$. 
    \item We now partition each cell $T_{v,e}$ into two cells (Figure~\ref{fig:integer}), as follows. 
    \begin{enumerate}
        \item If vertex $v$ is an end point of edge $e$, then the weight of the left and right cells are  $2$    and  $4$, respectively. We refer to these as a  \emph{$(2,4)$-group}.
        \item Otherwise, the weight of the left   and right cells are  $2$    and $2$, respectively. We refer to these as a  \emph{$(2,2)$-group}.
    \end{enumerate}
\end{enumerate}
It now suffices to show that $G$ admits a clique of size $k$ if and only if there exists a no-split StreamTable of height $H=(n+k)$ and width at most $  (6m - 2{k\choose 2})+(m+1)\epsilon$, 
where the row heights are integers.

Assume first that $G$ has a clique $C$ of size $k$. We then draw the StreamTable such that each row corresponding to the vertices of $C$ has a height of $2$ and every other row has a height of $1$   (Figure~\ref{fig:integer}). For each edge column $e=(v,w)$, there can now be two cases: (A) If $v\in C$ and $w\in C$, then $T_{v,e}$ and $T_{w,e}$ will be $(2,4)$ groups and all the other cells in this edge column are $(2,2)$ groups. Hence this edge column can be drawn with a width of $4$ units. (B) If at least one of $v$ and $w$ are not in $C$, then without loss of generality assume $w\not\in C$. Since $T_{w,e}$ is a $(2,4)$-group and the row height for $w$ is 1, we can draw the cells in  $T_{w,e}$ using a width of 6 units. It is straightforward to draw the remaining cells of this edge column within the same width. 

Since we have $k$ mutually adjacent vertices, we will have ${k \choose 2}$ edge columns with width $4$ and $(m-{k \choose 2})$ edge columns with width $6$. Thus the total width is at most $6(m-{k \choose 2})+4{k \choose 2} = (6m - 2{k\choose 2})$. Together with the line columns the width is at most $(6m - 2{k\choose 2})+(m+1)\epsilon$.
 \begin{figure}[pt]
    \centering
    \includegraphics[width=.8\textwidth]{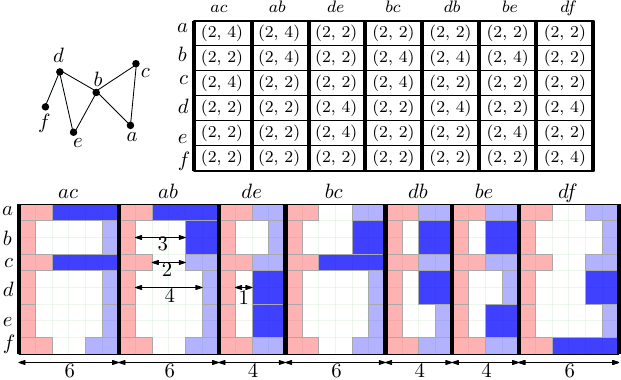}
    \caption{Illustration for the proof of Theorem~\ref{thm:integer}. Given a clique $\{b,e,d\}$, one can construct a StremTable with ${3\choose 2}$ edge  columns of width $4$, where the remaining edge columns are of width 6. The line columns are shown in thick vertical lines. }
    \label{fig:integer}
\end{figure}

Assume now that there exists a StreamTable for $T$ with height $(n+k)$ and width $(6m - 2{k\choose 2})+(m+1)\epsilon$, where all the row heights are integers. 

First assume that there are $k$ rows of height 2 or more. We now show that the corresponding vertices are mutually adjacent in $G$. We will use the notation $C'$ to denote these set of vertices.

Suppose for a contradiction that $C'$ does not determine  a clique. Let  $e=(v,w)$ be an edge column. Assume first that  $C'$ does not include both $v$ and $w$. Without loss of generality assume that $w\not \in C'$. Since the corresponding row is of height 1, the width required for the corresponding cell $T_{w,e}$ is at least 6 units. Assume now that  $C'$ includes both $v$ and $w$. Since  $k<n$, there must be at least one row of height 1. Let $z$ be the corresponding vertex. the width required for the corresponding cell $T_{z,e}$ is at least 4 units. If $C'$ does not correspond to a clique, then the total width of the table must be at least $6(m-m')+4m' - n(m+1)\epsilon$, where $m'$ is the number of edges of $G$ with both endpoints in $C'$ and the negative term $n(m+1)\epsilon$ compensate for the shift that may occur due to the line columns.

Since $\epsilon = \frac{1}{ n(m+1)}$, the width is at least  $6(m-m')+4m' - 1 = 6m-2m'-1$. Since $m'$ is smaller than ${k\choose 2}$, then $6m-2m'-1> 6m-2({k\choose 2}-1)-1 = 6m-2{k\choose 2}+1$. Since $ (m+1)\epsilon = \frac{1}{ n}$, we have $6m-2{k\choose 2}+1 > 6m-2{k\choose 2} + (m+1)\epsilon$. We thus reach a contradiction to our initial assumption on the width of $T$.

Assume now that there are less than $k$ rows of height 2 or more. Then the total width is at least  $6(m-m')+4m' - n(m+1)\epsilon$, where $m'$ is the number of edges of $G$ with both endpoints in $C'$. Since $m'$ is smaller than ${k\choose 2}$, we can apply the same argument as above  to reach a contradiction.
\end{proof}

\textbf{Geometric Programming:} 
\label{appgeom} 
Here we present a discussion on geometric programming that partially models the problem of Section~\ref{vrh}, which is inspired by the work on area minimization problem in floorplanning~\cite{DBLP:conf/ispd/ChenF98,DBLP:journals/mmor/Rosenberg89}. We show that if the cells are allowed to have an area larger than their corresponding weights, then the problem can be modelled using a geometric programming.  In particular, the optimization function and constraints of Section~\ref{sec:lp} can be expressed as follows:

\begin{equation*}
\begin{array}{lll@{}lll}
\text{minimize}  & \sum_{j=1}^r h_j\cdot b_{j,c} &\\
\text{s. t.}& a_{j,1} a_{j+1,1}^{-1} = 1, & \text{ }j=1 ,\dots, r-1 & \text{ }(C'_1)\\
                 & b_{j,c}b_{j+1,c}^{-1}=1, &\text{ }j=1 ,\dots, r-1 & \text{ }(C'_2)\\ 
                 & w_{j,k}h_j^{-1}b_{j,k}^{-1}  + a_{j,k}b_{j,k}^{-1} \le 1, &\text{ }j=1 ,\dots, r;\text{ } k=1 ,\dots, c  & \text{ }(C'_3)\\
                 &a_{j,k}b_{j,k}^{-1}\le 1, & \text{ }j=1 ,\dots, r;\text{ } k=1 ,\dots, c  & \text{ }(C'_4)\\
                 &a_{j,k}b_{j+1,k}^{-1}\le 1, & \text{ }j=1 ,\dots, r-1;\text{ } k=1 ,\dots, c  & \text{ }(C'_5)\\
                 &a_{j+1,k}b_{j,k}^{-1}\le 1, & \text{ }j=1 ,\dots, r-1;\text{ } k=1 ,\dots, c  & \text{ }(C'_6)\\
                 &h_1H^{-1}+\ldots+h_rH^{-1}\le 1, & \text{ }j=1 ,\dots, r;  & \text{ }(C'_7)\\
                 & W^{-1} b_{1,c} \le 1, & \text{ }j=1 ,\dots, r;  & \text{ }(C'_8)\\
                 & a_{j,k},b_{j,k},h_j\ge 0 & \text{ }j=1 ,\dots, r;\text{ } k=1 ,\dots, c  \\
\end{array}
\end{equation*}

Note that the objective function is minimizing the total area of the StreamTable, which is a posynomial function. In Section~\ref{sec:lp}, we discussed three constraints $C_1$--$C_3$. The constraints $C'_1$ and $C'_2$ follow from the constraint $C_1$, which has now  been rewritten with monomials.  The constraints $C'_3$ is similar to the constraint $C'_2$ except that instead of strict equality, we now have $h_jb_{j,k}-h_ja_{j,k} \ge w_{j,k}$, which has been rewritten with a posynomial function.  The constraints $C'_4$--$C'_6$ can be obtained from $C_3$, which has now been rewritten with monomials. Finally, the constraints $C'_7$ and $C'_8$ are added to keep the height and width bounded by $W$ and $H$, respectively.

This is well-known that a geometric programming can be transformed to a nonlinear but convex optimization problem via a logarithmic transformation of the objective and constraint functions~\cite{bb}. Although we have the constraints $C'_4$--$C'_6$ to ensure that the streams are connected, the cells themselves are now allowed to have an area larger than  their corresponding weights. Shrinking the cells to realize their original weights (by creating rectangular gaps) does not guarantee the connectedness of the streams. Hence the original problem remains open.

\section{StreamTable (Uniform Row Heights, Variable Row Order)}

We now show that computing StreamTables with no split (resp., minimum excess area) while minimizing the excess area (resp., number of splits) by reordering the rows  is NP-hard.

\begin{theorem}
\label{thm:ma}
Given a table $T$ and a non-zero positive number $\delta>0$, it is NP-hard to compute a StreamTable  with no split and minimum excess area, where each row is of height $\delta$ and the ordering of the rows can be chosen.
\end{theorem}
\begin{proof}
We reduce the NP-complete problem \emph{betweenness}~\cite{DBLP:journals/siamcomp/Opatrny79}, where the input is a set of ordered triples of elements, and the problem is to decide whether there exists a total order $\sigma$ of these  elements, with the property that for each given triple, the middle element in the triple appears somewhere in $\sigma$  between the other two elements.

Let $S$ be a set of $c$ integer triples (an instance of betweenness) over $r$ elements (integers), where $r,c\ge 5$. We now construct an $r\times (4c+1)$ table $T$  (Figure~\ref{fig:x}(a)), as follows: 

\begin{enumerate}
    \item For every triple $t \in S$, we make a column  (labelled with $t$). We refer to these columns as  \textit{triple columns}. Each of these columns will later be split into three more columns. For every element $e$, we create a row (labelled with $e$). Assume that each cell of a triple column has a weight of $w (=15)$. 
    \item We insert an additional column at the left and right sides of the table and also between every pair of adjacent triple columns.  We refer to these columns as \textit{line columns}. Each cell of a line column has a weight of $\epsilon = \frac{1}{ r(c+1)}$.  
    \item For every triple $t$ and row $i$, we further partition the cell $(t,i)$ into three cells and distribute the weight $w$ among them, as follows:
    \begin{enumerate}
        \item[4a] If $i$ is the left element of $t$, then the  weight of the left, middle and right cells are  $\frac{2w}{3}, \frac{w}{6}$ and $\frac{w}{6}$, respectively.
        \item[4b] If $i$ is the right element of $t$, then the  weight of the left, middle and right cells are  $\frac{w}{6}, \frac{w}{6}$ and $\frac{2w}{3}$, respectively.
        \item[4c] If $i$ is the centre element of $t$, then the  weight of the left, middle and right cells are  $\frac{w}{6}, \frac{2w}{3}$ and $\frac{w}{6}$, respectively. 
        \item[4d] Finally, if $i$ does not belong to $t$, then the  weight of the left, middle and right cells are  $\frac{5w}{12}, \frac{w}{6}$ and $\frac{5w}{12}$, respectively.  
    \end{enumerate}
\end{enumerate}

   
\begin{figure}[h]
    \centering
    \includegraphics[width=\textwidth]{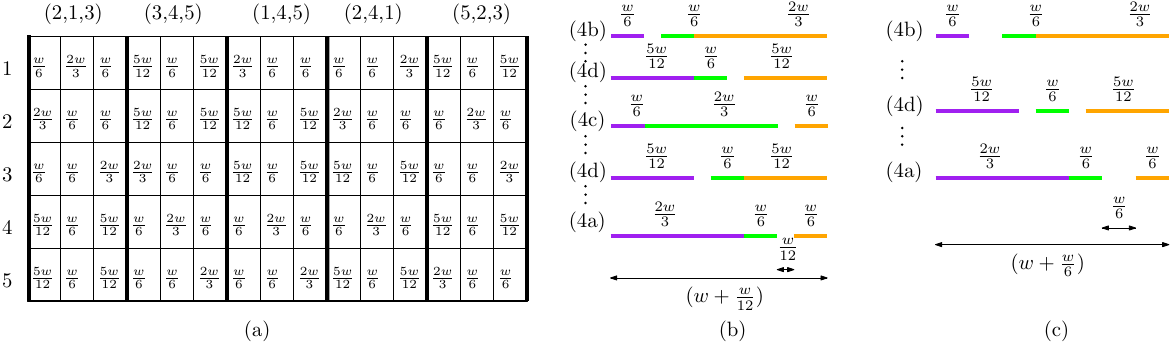}
    \caption{(a) A table $T$ obtained from a set of triples $\{(2,1,3),$ $ (3,4,5),$ $ (1,4,5), (2,4,1),(5,2,3)\}$. Here the thick black lines represent the line columns. (b)--(c) Illustration for the required width for different row orderings.}
    \label{fig:x}
\end{figure}


We set $\delta$ to be 1. It now suffices to show that  the betweenness instance $S$ admits a total order $\sigma$, if and only if there exists a StreamTable with no split and at most $\frac{rcw}{12}$ excess area, where each row is of height $\delta$. 

First assume that $S$ admits a total order $\sigma$. 
We draw the rectangles of each line column on top of each other (vertically aligned) and allocate a width of  $(w+\frac{w}{12})$  for the triple columns. Since we order the rows by $\sigma$, a pair of rows that satisfy conditions (4a) and (4b) for a triple $t$ must have a row $k$ satisfying condition (4c) for $t$. Therefore, we can complete the drawing of the rectangles of the three streams within the allocated width without any split. Figure~\ref{fig:x}(b) illustrates a schematic representation of the rows for this scenario. Since $\delta = 1$, the excess area is at most $ \frac{rcw}{12}$ in total. Figure~\ref{fig:y} illustrates the construction for the table from Figure~\ref{fig:x}(a).
\begin{figure}[h]
    \centering
    \includegraphics[width=.6\textwidth]{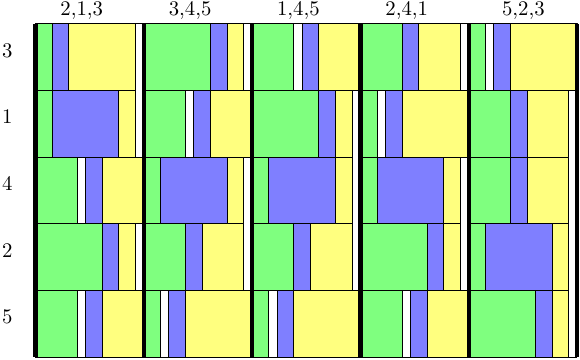}
    \caption{A StreamTable for the $T$, where $\sigma = \{3,1,4,2,5\}$.}
    \label{fig:y}
\end{figure}

We now show that if there is a StreamTable for $T$ with  at most $\frac{rcw}{12}$ excess area, then the corresponding row ordering will yield the total order for the betweenness instance. Any three streams corresponding to a triple $t$ must have exactly one $(4a)$, one $(4b)$ and one $(4c)$ conditions. Suppose for a contradiction that for some triple $t$, the condition $(4c)$ does not appear between conditions $(4a)$ and $(4b)$ (Figure~\ref{fig:x}(c)). Then these streams would require a width of at least $(w+\frac{w}{6})$. Let $t'$ be a triple that appears immediately after $t$. Since the stream (line column) between $t$ and $t'$ is very narrow, they cannot share a  width of more than $r\epsilon$. By construction, for each triple, we have three streams and they require a width of at least $(w+\frac{w}{12})$. Hence the total width of the visualization must be at least
\begin{align*}
 & (w+\frac{w}{6} ) + (c-1)(w+\frac{w}{12}) - r(c+1)\epsilon\\
&= cw + \frac{cw}{12} + \frac{w}{12} -  r(c+1)\epsilon\\
&= cw + \frac{cw}{12}+ \frac{w}{12} -1.
\end{align*}
Since the sum of the weights of $T$ is $(rcw+r(c+1)\epsilon) = (rcw+1)$, this implies an excess area of  larger than 
$\frac{rcw}{12} + \frac{rw}{12}-r -1> \frac{rcw}{12}$, which contradicts our assumption (when $w\ge 15$) that the excess area is at most $\frac{rcw}{12}$.  
 \end{proof}
 
 
\begin{theorem}
\label{thm:ms}
Given a table $T$ and a non-zero positive number $\delta>0$, it is NP-hard to compute a StreamTable  with zero excess area and minimum number of splits, where each row is of height $\delta$ and the ordering of the rows can be chosen.
\end{theorem} 
\begin{proof}
We reduce the NP-complete problem \emph{Hamiltonian path in a cubic graph}~\cite{DBLP:journals/siamcomp/GareyJT76}, where the input is a graph $G$ with $n$ vertices and $m$ edges such that every vertex is of degree 3, and  the problem is to decide whether there exists a total order  of the vertices that determines a Hamiltonian path in $G$. 
\begin{figure}[h]
    \centering
    \includegraphics[width=.75\textwidth]{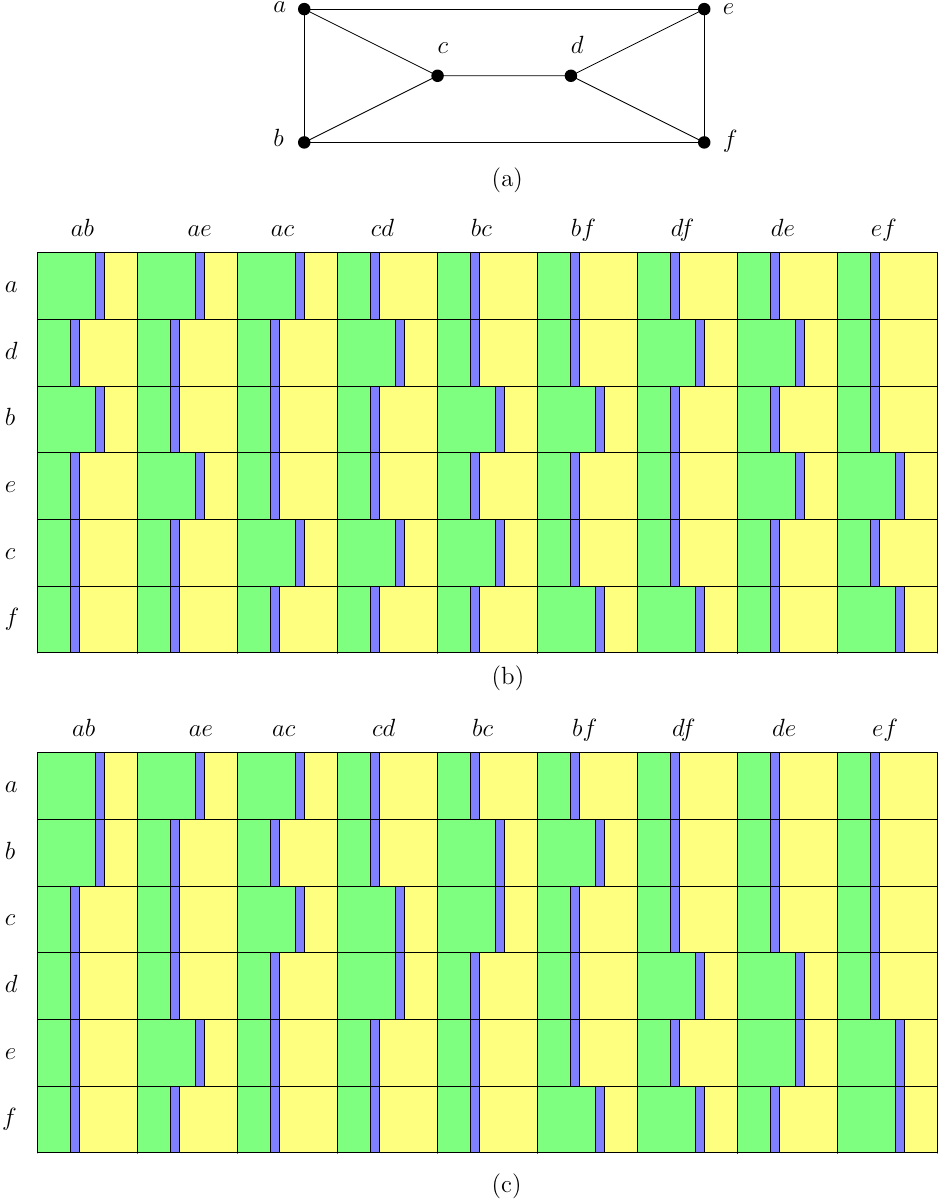}
    \caption{(a) A cubic graph $G$. (b) A visual representation of the table $T$ corresponding to $G$. (c) Construction of a StreamTable from a given Hamiltonian path $a,b,c,d,e,f$.}
    \label{fig:z}
\end{figure}

 Let $G$ be a graph with $n$ vertices, i.e., an instance of the Hamiltonian path problem. We now construct a table $T$, as follows. 

\begin{enumerate}
    \item For each edge $e \in E_{G}$, we create a column  called an \emph{edge column}, and label it by $e$.   
    \item For each vertex $v \in V_{G}$, we create a row and assign it the name $v$. For each edge column, we assign each cell a weight $w$.
    \item We now partition each cell $T_{v,e}$ into three cells (Figure~\ref{fig:z}(a)--(b)), as follows. 
    \begin{enumerate}
        \item If vertex $v$ is an end point of edge $e$, then the weight of the left, middle and right cells are  $\frac{7}{12}w$, $\frac{1}{12}w$, and  $\frac{4}{12}w$, respectively. We refer to these as an \emph{$L$ group}.
        \item Otherwise, the weight of the left, middle and right cells are  $\frac{4}{12}w$,  $\frac{1}{12}w$, and $\frac{7}{12}w$, respectively. We refer to these as an \emph{$R$ group}.
    \end{enumerate}
\end{enumerate}


It now suffices to show that $G$ admits a Hamiltonian path if and only if there exists a StreamTable for $T$ with zero excess area and at most $4(n-1)$ splits, where the height of each row is $\delta=1$.

Assume first that $G$ has a Hamiltonian path $P$. We then draw the StreamTable such that each row has a height of $\delta$, each row is drawn in the order of the Hamiltonian path, and the cells within each row are drawn consecutively without leaving any gap in between  (Figure~\ref{fig:z}(c)).  By construction of the StreamTable, for every pair of vertices $v,w$ which are adjacent in $P$, the corresponding rows will be consecutive in the StreamTable. For each edge column $e$, there can now be three cases: (A) If $e=(v,w)$, then $T_{v,e}$ and $T_{w,e}$ will consist of $L$ groups and hence no split will appear. (B) If neither $v$ nor $w$ is an endpoint of $e$, then $T_{v,e}$ and $T_{w,e}$ will consist of $R$ groups and hence no split will appear. (C) Otherwise, one of $T_{v,e}$ and $T_{w,e}$ will represent an $L$ group and the other is an $R$ group, and hence a split will appear. Therefore, the number of splits contributed by the rows representing $v$ and $w$ is $4$. Thus the number of splits overall is $4(n-1)$. 

Assume now that there exists a StreamTable for $T$ with zero excess area and at most $4(n-1)$ splits, where each row is of height $\delta$. We now show that the row ordering in the StreamTable  determines a Hamiltonian path in $G$. Suppose for a contradiction that there exists a consecutive pair of rows in the StreamTable where the corresponding vertices are not adjacent in $G$. Let $a,b$ be such a pair of vertices. Then every  $L$ group of $a$ will occupy a horizontal interval shared with an $R$ group of $b$, and vice versa. Hence this would contribute to at least $6$ splits. Since any pair of consecutive rows must contribute to at least 4 splits, the number of splits will be at least $6+4(n-2) = 4n-2 > 4(n-1)$. 
\end{proof}

\section{Conclusion}
In this paper we have introduced StreamTable, which is an area proportional visualization inspired by streamgraphs. We formulated  algorithmic problems that need to be tackled to produce  aesthetic StreamTables and examined two aesthetic criteria -- excess area and number of splits.  

We have showed that if row heights and row ordering are given, then a StreamTable with no split and minimum area can be computed via a linear program. However, the case when the row ordering is given but the row heights can be chosen needs further investigation. We only provided a quadratically-constrained quadratic program to model the problem and an NP-hardness proof when the row heights are constrained  to be integers. However the original question remains open.

\smallskip
\noindent\textbf{Open Problem 1:} Given a table $T$ and a positive integer $H$, does there exist a polynomial-time algorithm to compute a minimum-area no-split StreamTable of height $H$ that respects the row ordering of $T$? 
\smallskip

We also showed that if the row ordering can be chosen, then the problem of finding a minimum-area or a minimum-split StreamTable is NP-hard. In this setting, it would be interesting to find algorithms for computing   zero-excess-area (resp., no split) StreamTables with good approximation on the number of splits (resp.,  excess area). 

\smallskip
\noindent\textbf{Open Problem 2:} Design polynomial-time algorithms to find good approximation for StreamTable aesthetics (excess area or number of splits) in both the fixed and variable row ordering settings. 
\smallskip

 Recently a framework for $\exists\mathbb {R}$-completeness of packing problems has been proposed in ~\cite{DBLP:conf/focs/AbrahamsenMS20}. It would be interesting to investigate   $\exists\mathbb {R}$-completeness in this context, where the rows need to be packed inside a rectangle maintaining column adjacencies.
\bibliography{CONFERENCE/eurocg21}
\end{document}